\documentclass[letterpaper, 10 pt, conference]{ieeeconf}  

\IEEEoverridecommandlockouts                              

\usepackage{graphics} 
\usepackage{epsfig} 
\usepackage{mathptmx} 
\usepackage{times} 

\usepackage{amssymb}  

\usepackage{algorithm}
\usepackage{algpseudocode}
\usepackage{amsmath} 
\usepackage{centernot}
\usepackage{array}
\newcolumntype{P}[1]{>{\centering\arraybackslash}p{#1}}
\newcolumntype{M}[1]{>{\centering\arraybackslash}m{#1}}
\usepackage[normalem]{ulem}
\usepackage{rotating}
\usepackage{tabularray}
\usepackage{rotating}
\usepackage{amsfonts}
\usepackage{xcolor}
\usepackage{graphicx}
\usepackage{wrapfig}
\usepackage[utf8]{inputenc}

\usepackage{todonotes}
\usepackage{amsfonts, amssymb, bbm}
\usepackage{comment} 
\usepackage{xcolor}
\usepackage{tikz}
\usepackage{hyperref}
\usetikzlibrary{arrows,shapes,snakes,automata, calc, chains, matrix, positioning, scopes, arrows.meta}

\usepackage{tabularx}
\usepackage{subcaption} 
\usepackage{url}
\usepackage{xpatch}
\usepackage{lscape}
\usepackage[misc,geometry]{ifsym}

\newtheorem{definition}{Definition}
\newtheorem{lemma}{Lemma}
\newtheorem{remark}{Remark}

\newcommand{\sig}[1]{{{\small\textsf{{#1}}}}}
\newcommand{\txt}[1]{{\small{\texttt{{#1}}}}}
\newcommand{\CommentedText}[1]{}

\usepackage{color}
\definecolor{darkgreen}{rgb}{0,0.5,0}
\definecolor{purple}{rgb}{1,0,1}
\newcommand{\kibitz}[2]{\ifnum\Comments=1\textcolor{#1}{#2}\fi}

\newcommand\defeq{\mathrel{\stackrel{\makebox[0pt]{\mbox{\normalfont\scriptsize def}}}{:=}}}

\newtheorem{example}{Example}

\title{\LARGE \bf
Instance-Level Safety-Aware Fidelity of Synthetic Data and \\
Its Calibration
}

\author{Chih-Hong Cheng$^{1}$, Paul Stöckel$^{2}$ and Xingyu Zhao$^{3*}$ 
\thanks{$^{1}$Department of Computer Science and Engineering, 
Chalmers University of Technology, Sweden. The work was conducted during his service at Univ. Hildesheim. }%
\thanks{$^{2}$Department of Informatics, 
University of Hildesheim, Germany.}%
\thanks{$^{3}$WMG, University of Warwick, CV4 7AL, United Kingdom
       }%
\thanks{$^{*}$Correspondence to:  {\tt\small xingyu.zhao@warwick.ac.uk}\vspace{1mm}}
}

\begin{document}

\maketitle
\thispagestyle{empty}
\pagestyle{empty}

\begin{abstract}

Modeling and calibrating the fidelity of synthetic data is paramount in shaping the future of safe and reliable self-driving technology by offering a cost-effective and scalable alternative to real-world data collection. We focus on its role in safety-critical applications, introducing four types of instance-level fidelity that go beyond mere visual input characteristics. The aim is to ensure that applying testing on synthetic data can reveal real-world safety issues, and the absence of safety-critical issues when testing under synthetic data can provide a strong safety guarantee in real-world behavior. We suggest an optimization method to refine the synthetic data generator, reducing fidelity gaps identified by deep learning components. Experiments show this tuning enhances the correlation between safety-critical errors in synthetic and real data.

\end{abstract}


\section{Introduction}
\label{sec:intro}

The increasing reliance on synthetic data for virtual testing is crucial in solving contemporary engineering challenges. In the automotive sector, scenario-based virtual testing with synthetic data is now the norm for ensuring autonomous vehicle (AV) safety \cite{iso_scenario_2022,bsi_scenario_2023}, replicating real-world scenarios. Significant investments are made in developing advanced simulators~\cite{dosovitskiy2017carla,son2019simulation} and Generative AI (GAI)~\cite{zhong2023guided,hu2023gaia} models to improve aspects such as fidelity, creating realistic virtual environments and sensors that closely resemble real driving conditions. However, achieving perfect fidelity in synthetic data generation remains a challenge. This leads to the question driving our work: ``\emph{What level of fidelity is necessary for synthetic data to be deemed adequate for safety purposes?}''

In this paper, we focus on ``\emph{instance-level fidelity}''\footnote{In virtual testing for AVs, fidelity may also refer to the closeness of synthetic data \textit{distribution} to the real-world data distribution, and the accuracy of the \textit{sequential dynamics} as the AV controller interacts with the virtual environments. While those perspectives are important, this paper focuses on examining fidelity at the level of \textit{an individual synthetic data-point}, e.g., an image or a point-cloud depending on the type of sensors simulated; in this paper, we study synthetic images without loss of generality.}, a term we introduce to describe how accurately an individual synthetic data point, such as an image, replicates the real world's characteristics, features, and variability. We provide a mathematical framework for four different fidelity types, illustrated in Fig.~\ref{specturm_fidelity}. At the left end of this spectrum is ``Input Value fidelity'' (IV-fidelity), which requires pixel-level similarity between synthetic and real-world images. IV-fidelity is \emph{model-agnostic}, aiming for a universally high fidelity level suitable for any System-Under-Testing (SUT).
Despite the stringent criteria of IV-fidelity, it does not guarantee the sameness of output, mainly due to Deep Neural Networks' (DNN) lack of robustness~\cite{goodfellow2014explaining}. 
Therefore, it's more effective to \emph{evaluate fidelity by considering how the specific SUT processes synthetic data and its intended use}, rather than solely focusing on strict fidelity requirements.

Towards this direction, we thus introduce the fidelity from the restricted view using the \emph{safety perspective}, i.e., the ``\emph{Safety-Aware fidelity}'' (SA-fidelity) on the rightmost of Fig.~\ref{specturm_fidelity}. Conceptually, what differentiates a superior synthetic data generator of higher SA-fidelity from a normal one is the \emph{alignment of reproducing safety-critical events}. This implies the capability to generate input with objects (e.g., pedestrians) that can, when interpreting the semantics of the input and reconstructing its scenario in the real world, lead to similar safety-critical concerns. At the same time, the absence of safety-critical issues when testing using synthetic data can provide a strong safety guarantee in real-world behavior. For example, consider a scenario where a pedestrian is extremely far from the autonomous vehicle. Even when the synthetic image generator fails to produce the appropriate contour for the distant pedestrian in a 2D image plane due to the object being extremely small, there exist no safety concerns, as whether or not a DNN-based perception system detects the pedestrian at the moment does not matter. This ultimately leads to a \emph{pragmatic, risk-aware} definition of instance-level fidelity that does not demand \emph{unnecessary perfection}, such as pixel-level precision between synthetic and real images.

With fidelity being an SUT-specific measure, we further consider techniques to adjust the synthetic data generation process to increase the SA-fidelity. Classical synthetic data generators may have parameters such as noise or camera angle, while neural synthetic data generators have learnable parameters. These parameters can be further \emph{calibrated} according to the specific need. The calibration of the synthetic data generator can thus be viewed as an \emph{optimization problem} to reduce the discrepancy in SA-fidelity, where techniques such as genetic programming or gradient descent can be applied depending on the type of synthetic data generation. 

\begin{figure*}[t]
    \centering
\includegraphics[width=0.85\textwidth]{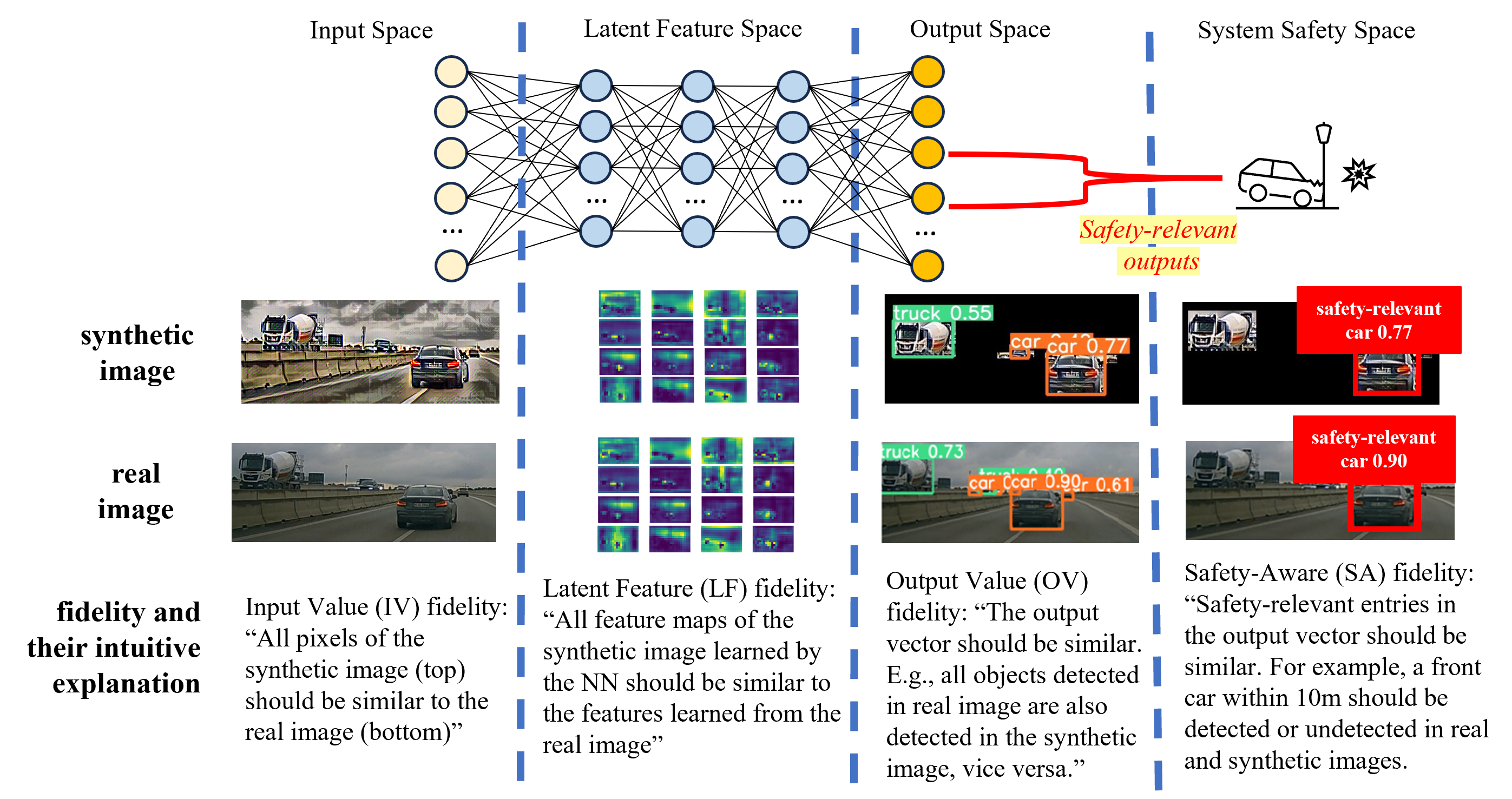}
 \caption{The spectrum of instance-level fidelity from the perspectives of DNN inputs, features extracted, DNN outputs and its effect on system safety.}
 \vspace{-5mm}
  \label{specturm_fidelity}

\end{figure*}

To validate the proposed concept of SA-fidelity and its calibration, we create four synthetic datasets aligned with the read-world vision-based object detection dataset KITTI~\cite{geiger2012we}, and then evaluate the SA-fidelity for each generator. Finally, for a given synthetic data generator, we calibrate some configurable parameters to achieve higher SA-fidelity. Our preliminary evaluation demonstrates that: 1) such a calibration surely leads to better synthetic images in terms of replicating more safety issues of the DNN as if it was deployed in the real world; 2) the optimal configuration for SA-fidelity is indeed not necessarily the same as the configurations for optimizing other types of fidelity. Such encouraging results reconfirm \textit{a practical need for a rigorous and task-oriented definition for the fidelity of synthetic data}.

In summary, the main contribution of this paper includes:

\begin{itemize}
    \item A comprehensive set of concepts on instance-level fidelity of synthetic data, with definitions, metrics, and formal analysis of their relationships.
    \item A first framework formulating the calibration of the synthetic data generation process for safety, with preliminary results showing its effectiveness.
    \item A public repository on all our synthetic datasets, codes, models, and experimental results at \url{https://github.com/semta-group/safe-syn-fidelity}.
\end{itemize}

The paper is structured as follows. After summarizing related work in Sec.~\ref{sec:related.work} and building the preliminaries in Sec.~\ref{sec:foundations}, we introduce three fidelity concepts, their relations, and the extension of SA-fidelity in Sec.~\ref{sec:facets.instance.level.fidelity}. 
We then present an optimization-based approach for post-calibrating the synthetic data generator and the evaluation in Sec.~\ref{sec:calibration} and~\ref{sec:evaluation}. Finally, we conclude with future directions in Sec.~\ref{sec:conclusion}.

\section{Related Work}
\label{sec:related.work}

Simulation-based testing for autonomous systems has been an active field with accessible commercial tools. Testing non-perception modules is relatively mature (see~\cite{zhong2021survey} for a review), where the scenario description can be translated into objects (as bounding boxes) and their physical characteristics (e.g., velocity, coordinate), serving as the input to the motion planner. The fidelity of cyber-physical system simulation (with differential equation included) has been addressed in various frameworks such as Ptolemy II~\cite{lee2015modeling},  Modellica~\cite{mattsson1998physical,lindqvist2023modelica}, and Mathworks Simscape\footnote{\url{https://www.mathworks.com/products/simscape.html}}. 
Our paper primarily focuses on image-based perception; simulation tools such as NVIDIA DriveSim~\cite{DriveSim} can create ``photo-realistic images'' as the underlying pipelines supporting physics-based rendering. These efforts lead to synthetic datasets like SynPeDS~\cite{stauner2022synpeds} tailored for pedestrian detection tasks in urban scenes. With a virtual world modeled, varying scene parameters can reveal performance-limiting factors in perception models~\cite{grau2022variational}. 
Apart from physic-based synthetic data generation, the use of GAI also hints at strong promises~\cite{zhong2023guided,hu2023gaia}. Despite the aforementioned progress being made, fundamental concerns remain regarding the validity of using virtual testing as a strong safety argument (i.e., to replace ``real miles'' with ``virtual miles''~\cite{haq2021can}). This motivates us to introduce a novel perspective to existing fidelity taxonomies (e.g., \cite{9706219}), and mathematically precise the meaning of fidelity by aligning it with its intended use, while SA-fidelity provides a rigorous statement to avoid over-perfecting every detail in simulation like the effort in~\cite{wallace2022validating}. 

Within the ML community, the technique of domain adaptation (see~\cite{wang2018deep} for an extensive survey) allows a DNN to perform well on new input domains via techniques such as training the DNN to encourage domain confusion in the feature space. Practically, one can utilize the concept of domain adaption, leading to training with both real and synthetic data included (e.g.,~\cite{tranheden2021dacs} on mixing Cityscapes (real) and GTA5 (synthetic)), thereby encouraging a DNN have consistent performance in both real and synthetic images. Our result of SA-fidelity complements these results, as \textit{1)} we enable a safety-oriented, task-specific assessment of the synthetic data (i.e., not on the DNN), and \textit{2)} our calibration on the synthetic data generator allows the DNN to remain unmodified, thereby acting on top of any existing domain adaptation scheme. 

 Finally, our motivation for SA-fidelity is based on the observation that the prediction inconsistency of very distant objects between synthetic and real images should not be an issue. While we aim to evaluate synthetic data fidelity, the concept of differentiating object detection performance conditional to the distance to the ego vehicle is also reflected in recent results~\cite{lyssenko2022towards,hung2022let,deng2021revisiting} on designing new evaluation metrics and training methods for 3D detection.

\section{Foundations}
\label{sec:foundations}

Denote $\mathbb{A}$ as the set of \emph{semantic attributes} (e.g., moving and scenery objects) that can be characterized by a \emph{scenario description language}. Without loss of generality (WLOG), we assume that each semantic attribute $a \in \mathbb{A}$ is a parameter that takes on a real number\footnote{Other attribute types can be encoded as real numbers.} as its value. A \emph{scenario description} $sd=\{(a_1,v_1), \dots, (a_k,v_k) \}$ contains a subset of $k$ attributes $\{a_1, \ldots, a_k\} \subset \mathbb{A}$ and associates each attribute with a concrete value $v_i, i\in [1..k]$.

Let $x \in \mathbb{R}^{d_{0}}$  be a data point (image, lidar point-cloud) with $d_{0}$ being its dimension. A synthetic data generator $g_\theta$ takes a scenario description $sd$ and creates a data point $x^s \defeq g_\theta(sd) \in \mathbb{R}^{d_{0}}$; $g_\theta(sd)$ can either be deterministic or stochastic, whereby calling~$g_{\theta}$ stochastic, we refer to the situation that $g_\theta(sd)$ may not generate the same results due to sampling. 
The parameter $\theta$ denotes configurable parameters of the synthetic data generator that will normally not be described by scenario description languages, e.g., the virtual sensor settings in simulators and the guidance scale setting in GAI models. Given a scenario description~$sd$, let $\sig{RW}(sd) \subseteq \mathbb{R}^{d_0}$ be an oracle function that returns an infinite set of \emph{real-world} data points whose semantic attributes match that of~$sd$. A fidelity concern arises because the images generated by the synthetic data generator, given a scenario description $sd$, may not be strictly contained in $\sig{RW}(sd)$.

\vspace{1mm}
\begin{remark}[\textbf{Approximation of $\emph{\sig{RW}}(sd)$}]
\label{remark_rw_approximation}
    Given a scenario description $sd$, $\sig{RW}(sd)$ is infinite and thus unknown in theory, which can only be approximated in practice, e.g., from a collected real-world dataset $X^{cr} \subset \mathbb{R}^{d_0}$. This leads to an approximated set $\widehat{\sig{RW}}(sd)\defeq \{x \mid x \in X^{cr} \cap \sig{RW}(sd)\}$.
\end{remark}
\vspace{1mm}

Let $f: \mathbb{R}^{d_0} \rightarrow \mathbb{R}^{d_L} $ be a perception function that takes an input data point $x \in \mathbb{R}^{d_0}$ and produces an output vector of $d_L$ dimensions. Consequently, $f$ can accept an input data point $x^s$ taken from $g_{\theta}(sd)$ as well as any data point within~$\sig{RW}(sd)$. When $f$ is a DNN, we use $f^l(x)$ to represent the feature vector of~$f$ at the $l$-th layer taking input~$x$.

\section{The Facets of Instance-level Fidelity}
\label{sec:facets.instance.level.fidelity}

This section starts with definitions of fidelity over input, output, and latent features. The definition enables us to pinpoint safety concerns regarding simulation fidelity rigorously. We then precise the concept of SA-fidelity as a relaxation on output value fidelity, where examples are introduced to assist understanding. Note that for simplicity, our definition is only \emph{one-sided} by demanding reproducibility of synthetic inputs (regarding input value, generated features, or output value) in the real world.

\subsection{Input, output, and latent feature fidelity}

\begin{definition}[\textbf{IV-fidelity}] A synthetic data 
$x^s \in g_{\theta}(sd)$ is defined to be $\langle \mathcal{D}_{in}, \epsilon_{in} \rangle$ Input-Value (IV) fidelitous iff
\begin{align}
\label{eq:iv.fidelity}
\exists \tilde{x}\in \sig{RW}(sd):\mathcal{D}_{in}(x^s,\tilde{x})\leq \epsilon_{in} 
\end{align}
where $\mathcal{D}_{in}$ is a distance metric between vectors and~$\epsilon_{in} \in \mathbb{R}_{+}$ is positive value.
\end{definition}

\vspace{1mm}
\begin{definition}[\textbf{OV-fidelity}] Given $f$, a synthetic data point $x^s\in g_{\theta}(sd)$ is defined to be $\langle \mathcal{D}_{out},\epsilon_{out}\rangle$ Output-Value (OV) fidelitous iff
\begin{align}
\label{eq:ov.fidelity}
\exists \tilde{x}\in \sig{RW}(sd): \mathcal{D}_{out}(f(x^s), f(\tilde{x})) \leq \epsilon_{out} 
\end{align}
where $\mathcal{D}_{out}$ is a distance metric between vectors and~$\epsilon_{out} \in \mathbb{R}_{+}$ is positive value.
\end{definition}

Intuitively, IV and OV-fidelity are defined based on some distance metrics between the given synthetic data and \textit{any} real-world data that sharing the same $sd$, before and after being processed by the DNN $f$.

\vspace{1mm}
\begin{remark}[\textbf{Conservative bound on distance-based fidelity}]
\label{remark_rw_bound}
As per Remark~\ref{remark_rw_approximation}, we can only implement the approximated $\widehat{\sig{RW}}(sd)$ when measuring IV/OV-fidelity in practice. Using the IV-fidelity as an example, since apparently $\min\{\mathcal{D}_{in}(x^s,\tilde{x})\mid \tilde{x} \in \widehat{\sig{RW}}(sd) \} 
\geq
\min\{\mathcal{D}_{in}(x^s,\tilde{x})\mid \tilde{x} \in \sig{RW}(sd) \}$, the practical approximation is calculating an upper bound on the fidelity discrepancy distance (similar reasoning holds for OV-fidelity). Thus, we may take such conservative approach that errs on the side of caution to prioritize safety.
\end{remark}
\vspace{1mm}

The following result connects IV-fidelity, OV-fidelity, and the common definition of \emph{local robustness} of DNNs \cite{huang2020survey}. 

\vspace{1mm}
\begin{lemma}[\textbf{IV-fidelity, robustness, and OV-fidelity}]\label{lemma:iv.robustness.ov}
Given a synthetic data point $x^s = g_{\theta}(sd)$, assume the perception module $f$ is locally robust with a maximum robustness bound $\epsilon_{r}$, i.e., 
$$\forall x': \mathcal{D}_{in}(x^s, x') \leq  \epsilon_{r} \rightarrow   \mathcal{D}_{out}(f(x^s), f(x')) \leq \epsilon_{out} $$
then if a synthetic data point $x^s$ is  IV-fidelitous under $\langle \mathcal{D}_{in},\epsilon_{in} \rangle$ where $\epsilon_{in} \leq \epsilon_{r}$, 
$x^s$ is also OV-fidelitous under $\langle \mathcal{D}_{out},\epsilon_{out} \rangle$.
\end{lemma}

\begin{proof} 
The lemma holds simply due to applying the transitive rule: the robustness radius~$\epsilon_r$ (where OV-fidelity guarantees to hold) is no less than~$\epsilon_{in}$. 
\end{proof}

\vspace{1mm}
\begin{remark}[\textbf{Photo-realistic is not real}] 

Lemma~\ref{lemma:iv.robustness.ov}~enables us to build an instance-level interpretation on the jargon ``photo-realistic is not real'', reflecting the common concern on the validity of using photo-realistic simulation in autonomous driving: The phrase ``photo-realistic'' can be interpreted as a synthetic input~$x^s$ being IV-fidelitous, while the phrase ``is not real'' reflects the issue that under the processing of the DNN, the real and the synthetic image may not lead to the sameness of output, i.e., OV-fidelity can not be guaranteed. This is analogous to adversarial examples~\cite{goodfellow2014explaining}, where OV-fidelity does not hold despite IV-fidelity may hold, reflecting the differences between the perturbed and the original image being human-eye indistinguishable. Altogether, one shall rephrase the slang to  ``\emph{photo-realistic is guaranteed to be real only when the difference between pixels is within the local robustness bound}''. 

\end{remark}

Finally, apart from the fidelity manifested in the input or output values, one can also define the fidelity based on the sameness of the learnt features between synthetic and real-world images, as demonstrated below. The sameness of the learnt features hints at the sameness of the decision mechanism, but again, analogously to the proof of Lemma~\ref{lemma:iv.robustness.ov}, one can also formally deduce that LF-fidelity does not immediately imply OV-fidelity. 

\vspace{1mm}
\begin{definition}[\textbf{LF-fidelity}] Given $f$ being a DNN with layers indexed from~$1$ to~$L$, a synthetic data point $x^s\in g_{\theta}(sd)$  is defined to be  $\langle \mathcal{D}_{lf}, \epsilon_{lf}, \mathcal{L}_{lf}\rangle$ Latent-Feature (LF) fidelitous subject to $f$, iff 
\begin{equation}
\label{eq:lf.fidelity}
  \begin{aligned}
\exists \tilde{x}\in \sig{RW}(sd): \; \forall l \in \mathcal{L}_{lf}:    \mathcal{D}_{lf}(f^l(x^s),f^l(\tilde{x}))\leq \epsilon_{lf} 
\end{aligned}
\end{equation}
where $\mathcal{D}_{lf}$ is a distance metric between vectors, $\epsilon_{lf} \in \mathbb{R}_{+}$ is a positive value, and $\mathcal{L}_{lf} \subseteq \{1, \ldots, L-1\}$.
\end{definition}
Notably, Remark \ref{remark_rw_bound} also holds for LF-fidelity.

\subsection{Safety-aware fidelity}

Prior to technical definitions, we first explain the design intuition of SA-fidelity. When applying DNN~$f$ with two data points~$x^s$ and~$x'$, the meaning of (output prediction) SA-fidelity $f(x) \equiv_{s} f(x')$ enables characterizing the behavior that both DNN output predictions do not lead to \emph{inconsistent safety concerns}.

\begin{itemize}
    \item Consider the case where the synthetic image $x^s$ and the real image $x'$ are almost identical, but the limitation of the synthetic image generator is that it can not generate perfect contours for very distant objects (e.g., pedestrians~$200$ meters away in low-speed driving scenarios) that have no safety concerns.  Therefore, despite the fact that $x^s$ may not satisfy IV or OV-fidelity, it shall be acceptable under the definition of SA-fidelity. 

    \item For a synthetic image $x^s$ being SA-fidelitous, it is possible that both $f(x^s)$ and $f(x')$ make similar safety-critical mistakes (e.g., miss a pedestrian in the front), thereby triggering the same type of safety concerns. The fidelity focuses on the \emph{consistency} of safety concerns between synthetic and real data points sharing the \emph{same scenario description}.

\end{itemize}

Recall that scenario description~$sd=\{(a_1,v_1), \dots, (a_k,v_k) \}$ contains a subset of $k$ attributes $a_1, \ldots, a_k \in \mathbb{A}$ and associates each attribute with a concrete value. Let \textbf{prediction influencing attributes} $\sig{pia}(sd)$ be a subset of a scenario description $sd$, where the attributes may be perceived by the perception unit. Let \textbf{safety influencing attributes}~$\sig{sia}(sd)$ return a subset of prediction influencing attributes~$\sig{pia}(sd)$, where an incorrect perception of these attributes can lead to safety concerns. 

\vspace{1mm}
\begin{example}\label{ex:pia.sa.car}
Consider a perception unit performing 3D object detection on cars for automatic emergency braking (AEB), and consider a scenario where there are front cars, one close and one distant to the autonomous vehicle, in rainy situations. The scenario description can be $sd \defeq\{(\sig{frontcar},1), (\sig{farcar}, 1),\dots, (\sig{rain},1) \}$, where the semantic attribute ``\sig{frontcar}'' indicates the existence of front vehicles being close-by (subject to a predefined physical distance),  ``\sig{farcar}'' indicates the existence of distant vehicles, and ``\sig{rain}'' indicates if the scenario is rainy. The $\sig{pia}(sd)$ and $\sig{sia}(sd)$ may act as follows. 

\begin{itemize}
    \item $\sig{pia}(sd) = \{(\sig{frontcar},1), (\sig{farcar}, 1)\}$, as the perception unit only infers knowledge on the existence of vehicles.
    \item $\sig{sia}(sd) = \{(\sig{frontcar},1)\}$, due to distant cars having no safety concerns. 
\end{itemize}
 
\end{example}

\vspace{1mm}
Subsequently, let $x \in \{g_{\theta}(sd)\} \cup \sig{RW}(sd)$, and define $\sig{inter}(x, f, sd)$ be an \textbf{interpretation function} that  ``rewrites'' the DNN output vector~$f(x)$ into a scenario description format, by changing $\sig{pia}(sd)$ and keep the rest the same as that of $sd$, due to DNN unable to recognize semantic attributes within $sd \setminus \sig{pia}(sd)$. 

\vspace{1mm}
\setcounter{example}{0}
\begin{example}[continued]
Consider again the scenario description $sd \defeq \{(\sig{frontcar},1), (\sig{farcar}, 1),\dots, (\sig{rain},1) \}$, where for an input $x \in \sig{RW}(sd)$, if the prediction~$f(x)$ fails to identify any vehicle, then 
$$
    \sig{inter}(x, f, sd) = \{(\sig{frontcar},0), (\sig{farcar}, 0),\dots, (\sig{rain},1) \}
$$
\end{example}

\vspace{1mm}
We can now formally characterize the meaning of two output predictions generating the same safety concerns (Def.~\ref{def:output.safety.similarity}) and the notion of SA-fidelity (Def.~\ref{def:sa.fidelity}). 

\vspace{1mm}
\begin{definition}[\textbf{Output safety similarity}]\label{def:output.safety.similarity}
Given scenario description~$sd$, DNN $f$, and input data points $x$ and $x'$ where $x, x' \in \{g_{\theta}(sd)\} \cup \sig{RW}(sd)$. The output safety similarity $f(x) \equiv_s f(x')$ holds iff the following condition holds.
\begin{equation}
\sig{sia}(\sig{inter}(x, f, sd)) = \sig{sia}(\sig{inter}(x', f, sd))
\end{equation}

\end{definition}

\vspace{1mm}

\setcounter{example}{0}
\begin{example}[continued] For an input $x' \in \sig{RW}(sd)$, if $f(x')$ fails to identify any vehicle, but for the synthetic data point $x^s = g_{\theta}(sd)$, $f(x^s)$ only fails to identify any far-distant vehicle. Then $f(x^s) \not\equiv_s f(x')$ due to the following  reasoning:
\begin{align}
\sig{inter}(x', f, sd) &= \{(\sig{frontcar},0), (\sig{farcar}, 0),\dots, (\sig{rain},1) \}& \nonumber \\
\sig{inter}(x^s, f, sd) &= \{(\sig{frontcar},1), (\sig{farcar}, 0),\dots, (\sig{rain},1) \}& \nonumber \\
    \sig{sia}(\sig{inter}(x', f, sd)) &= \{(\sig{frontcar},0)\} \nonumber \\
    &\neq \sig{sia}(\sig{inter}(x^s, f, sd))  = \{(\sig{frontcar},1)\}  \nonumber
\end{align}

\end{example}

\vspace{1mm}
\begin{definition}[\textbf{SA-fidelity}]\label{def:sa.fidelity} Given $f$, a synthetic data point $x^s\in g_{\theta}(sd)$ is defined to be Safety-Aware (SA) fidelitous under $\emph{\sig{inter}}(\cdot)$ and $\emph{\sig{sia}}(\cdot)$ iff
\begin{align}
\label{eq:sa.fidelity}
\exists \tilde{x}\in \sig{RW}(sd): f(x^s) \equiv_s f(\tilde{x})
\end{align}
\end{definition}
\vspace{2mm}

Note, aforementioned remarks on approximating $\sig{RW}(sd)$ and yielding conservative measurements holds for SA-fidelity.

\vspace{1mm}

\section{Calibrating Data Generation for Safety-aware Fidelity}\label{sec:calibration}

In the following, we present an optimization-based approach for fine-tuning the data generation process towards increasing safety-aware fidelity; the developed algorithm is equally applicable to all other types of fidelity criteria due to the criteria being mathematically defined.  

\vspace{1mm}

\begin{definition}[\textbf{Calibrator function}]
    Let the synthetic data generator  $g_{\theta} \defeq \mathcal{C}_{\theta} \circ \mathcal{G}$ be a composition of two functions~$\mathcal{C}_{\theta}$ and~$\mathcal{G}$, with the following characteristics: Given a scenario description $sd$, $g_{\theta}(sd) = \mathcal{C}_{\theta}(\mathcal{G}(sd))$, where
\begin{itemize}
    \item $\mathcal{G}(sd) \in \mathbb{R}^{d_0}$ is a synthetic data generator that can generate a legitimate input for the DNN~$f$. 
    \item $\mathcal{C}_{\theta}: \mathbb{R}^{d_0} \rightarrow \mathbb{R}^{d_0}$ is a \textbf{calibrator function} with $\theta$ being learnable parameters; it performs further transformation of an input data point while maintaining the dimension. 
\end{itemize}
\end{definition}

The formulation $g_{\theta} \defeq \mathcal{C}_{\theta} \circ \mathcal{G}$  reflects the practical need where $\mathcal{G}$ can be a commercial synthetic data generator whose underlying mechanism is unknown, acting like a black box. Given a scenario description~$sd$, a data point directly generated from $\mathcal{G}(sd)$ may have SA-fidelity issues. We aim to find the best $\mathcal{C}_{\theta}$, via optimizing the set of parameters~$\theta$ to reduce the inconsistency of safety concerns. This leads to the following optimization problem as formulated in Def.~\ref{def:calibration.sa.fidelity}, enabling the use of any optimization algorithm ranging from grid-based search to genetic programming.

\vspace{1mm}
\begin{definition}[\textbf{Calibration for SA-fidelity}]\label{def:calibration.sa.fidelity} 
Let~$\mathcal{D} \defeq \{(x, \mathcal{G}(sd))\}$ be the calibration set, containing pairs of real and synthetic data points, where for each pair $(x, \mathcal{G}(sd))$, the real data point $x$ has the scenario description being~$sd$, i.e., $x \in \emph{\sig{RW}}(sd)$. Eq.~\eqref{eq:calibration} formulates the design of the calibrator function as a minimization problem
{\small
\begin{equation}\label{eq:calibration}
    \sig{argmin}_{\theta} \!\!\! \!\!\! \!\!\! \!\!\! \sum_{(x, \mathcal{G}(sd)) \in \mathcal{D}} \!\!\! \!\!\!\!\!\!\!\!\!\!\!\sig{loss}( \sig{sia}(\sig{inter}(x, f, sd)) , \sig{sia}(\sig{inter}(\mathcal{C}_{\theta}(\mathcal{G}(sd)), f, sd)) )
\end{equation}}
where $\sig{loss}(x, y)$ is a function computing the difference between~$x$ and~$y$.

\end{definition}

\vspace{1mm}

In Eq.~\eqref{eq:calibration}, given a synthetic data point~$\mathcal{G}(sd)$, $\mathcal{C}_{\theta}(\mathcal{G}(sd))$ applies the transformation to create a modified data point. One can then compare the real data point~$x$ and the modified synthetic data point~$\mathcal{C}_{\theta}(\mathcal{G}(sd))$ in terms of output safety similarity (cf Def.~\ref{def:output.safety.similarity}). Finally, the design of function \( \sig{loss}(\cdot, \cdot) \) used in Eq.~\eqref{eq:calibration} is subject to user intention. One option is to use \( \sig{neq}(\cdot, \cdot) \) formally defined in Eq.~\eqref{eq:neq}: It evaluates instance-level SA-fidelity in a \emph{qualitative/Boolean} manner, by setting the value to~$1$ if output safety similarity is violated, and~$0$ if output safety similarity holds.  Another option is to use \emph{quantitative} $L_1$ loss following the definition of Eq.~\eqref{eq:L1.norm.semantic.attribute}, summing up the absolute difference of semantic attributes. Other options, such as $L_2$ or $L_{\infty}$ loss, can also be used.  
\begin{equation}\label{eq:neq}
            \sig{neq}(x, y) \defeq 
        \begin{cases} 
        1 & \text{if } x \neq y, \\
        0 & \text{otherwise.}
        \end{cases}
        \end{equation}
        
\begin{align}\label{eq:L1.norm.semantic.attribute}
\sig{diff}_{L1}(\{(a_1, v_1) ,\ldots, (a_k, v_k)\}, \{(a_1, v'_1) ,\ldots, (a_k, v'_k)\}) \nonumber\\ \defeq \sum^k_{i=1}|v_i - v'_i|
\end{align}

\section{Evaluation}
\label{sec:evaluation}

\subsection{Research questions}

We investigate three research questions (RQs)---RQ1 and RQ2 are designed to \textit{conceptually} validate the proposed SA-fidelity and its calibration, while RQ3 aims to understand its use in the \textit{practical} context of scenario-based AV testing.

\vspace{1mm}
\paragraph{\textbf{RQ1 (assessment)}: Can the proposed SA-fidelity metric effectively differentiate and rank different synthetic data generators in terms of safety?} Given a set of simulators/GAI-models, one normally wants to select the best solution for synthetic data generation, considering the computational budget in practice. Thus, it is important to examine whether the proposed SA-fidelity metric is capable of evaluating and comparing these generators.

\vspace{1mm}
\paragraph{\textbf{RQ2 (calibration)}: For a given synthetic data generator, can SA-fidelity be utilized for calibrating optimal configurations geared toward safety?} This RQ explores the potential of SA-fidelity metrics not only as assessment tools but also as crucial guides in the fine-tuning process of synthetic data generators. 

\vspace{1mm}
\paragraph{\textbf{RQ3 (integration with scenario-based testing)}: When integrating SA-fidelity concepts into the engineering process of scenario-based virtual testing, what are additional challenges to be considered?} This question comes from additional concerns due to the implementation steps for RQ1 and RQ2 do not fully align with the established engineering process of scenario-based virtual testing.

\subsection{Experiment setup}

\begin{figure*}[t]
    \centering
    \includegraphics[width=\textwidth]{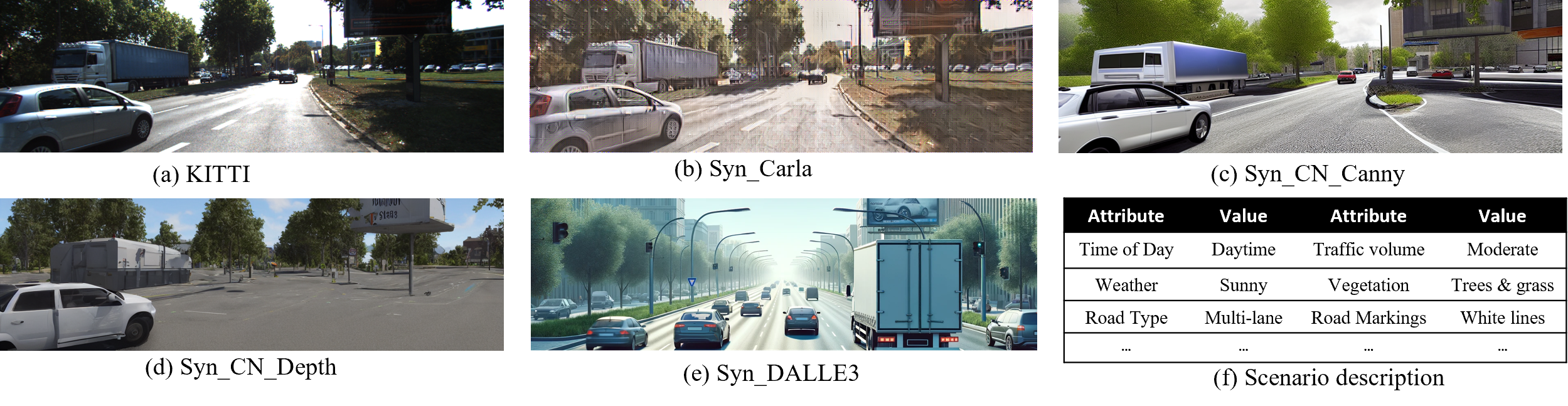}
 \caption{Example of a real-world KITTI image (a), and its synthetic images generated via  ``style transfer'' (b-d) and ``scenario understanding'' (e).}
  \label{fig_datasets_example}
  \vspace{-5mm}
\end{figure*}

\subsubsection{Synthetic data generation}

We generated four synthetic datasets based on the real-world AV dataset KITTI~\cite{geiger2012we}. Instead of utilizing all data in KITTI, in this initial study, we manually filtered out repetitive images (with similar semantic attributes) and randomly selected 90 real images from KITTI. Subsequently, we generated 90 corresponding images for each synthetic dataset, with examples shown in Fig.~\ref{fig_datasets_example}. The four datasets are generated in two different ways:
\begin{itemize}
    \item (``Style transfer'') Directly conduct style transfer from real to synthetic images using Neural Style Transfer\footnote{We used the solution from Phosus (\url{https://phosus.com/}) to generate images directly without our own programming efforts; neural style transfer is also supported by Keras (\url{https://keras.io/examples/generative/neural_style_transfer/}).} and ControlNet~\cite{zhang2023adding}. To be exact, \txt{Syn\_Carla} is the dataset generated by Neural Style Transfer taking the Carla simulated images as input; while \txt{Syn\_CN\_Canny} and \txt{Syn\_CN\_Depth} are two datasets by applying ControlNet with canny edge detection and depth map on the real KITTI images, respectively. The style transfer approach is ideally most suitable for rendering images from gamified simulators' 3D models, where in our setup, we particularly use the opposite direction to ensure the semantic sameness of the real-synthetic image pair. The generated datasets are used for accessing~RQ1 and ~RQ2.

    \item (``Scenario understanding'') Towards accessing RQ3, we additionally build a pipeline to query ChatGPT\footnote{In our evaluation, we used ChatGPT version~4.
    } to generate scenario descriptions (first as a paragraph, followed by summarizing into the form like Fig.~\ref{fig_datasets_example}(f)) on images from KITTI; then use DALLE-3~\footnote{DALLE-3: \url{https://openai.com/dall-e-3}} to synthesize new images according to the scenario description. Then manually label the classes and bounding boxes for objects in generated images. This results in the \txt{Syn\_DALLE3} dataset. 
\end{itemize}

\subsubsection{DNNs under testing}
Without loss of generality, we evaluate the methodology on three pre-trained object detection models available in PyTorch, namely \txt{ssd300\_vgg16} (the SSD300 model in~\cite{liu2016ssd}), \txt{fasterrcnn\_resnet50} (the Faster R-CNN model~\cite{li2111benchmarking} with a ResNet-50-FPN backbone) and \txt{fcos\_resnet50} (the FCOS model~\cite{tian2019fcos} with a ResNet-50-FPN backbone).

\subsection{Results}

\subsubsection{RQ1}
To measure SA-fidelity, we must first determine safety-relevant objects in a given real image. This is an application-specific task requiring case-by-case safety analysis for the SUT. For simplicity, we calculate the area size of detected objects and set a threshold in our experiments. We then determine an object to be safety-irrelevant if its size is smaller than this threshold, assuming it is too distant to pose any safety concerns. Certainly, this approach is an oversimplified substitute for a comprehensive safety analysis, but we believe it suffices to demonstrate and validate the \textit{conceptual} idea of SA-fidelity. E.g., in the real image of Fig.~\ref{fig_inconsistancy_number}(a), those cars in orange bounding boxes are too far away and thus labeled as ``don't care'', while only six near cars in green bounding boxes are labeled as ``car'' and need to be correctly recognized for safety.

\begin{figure}[h!]
    \centering
     \includegraphics[width=0.9\columnwidth]{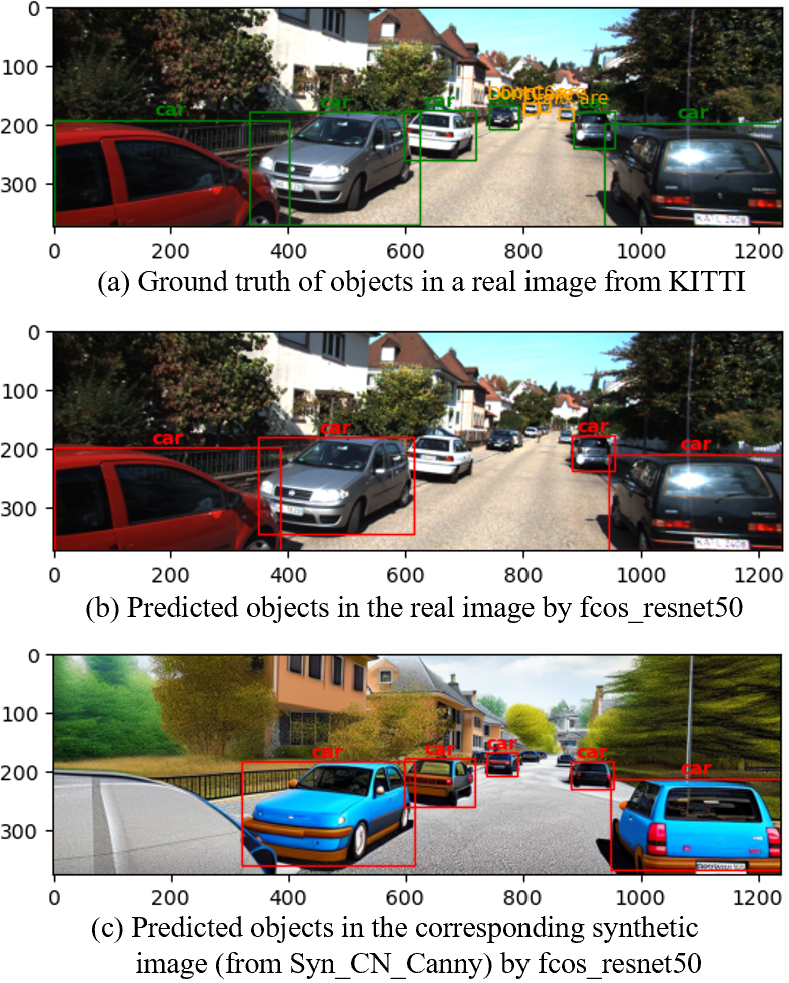}
 \caption{Examples of inconsistent object detection on safety-relevant objects in real vs. synthetic images.}
 \vspace{-5mm}
  \label{fig_inconsistancy_number}
\end{figure}

In Fig.~\ref{fig_inconsistancy_number}(b), the DNN detects only 4 out of the 6 safety-relevant cars, missing the two distant cars on the left. If used in the real world, this oversight would pose safety concerns and the virtual testing should flag such issues (i.e., the two cars should also be missed in virtual testing). However, in the synthetic image of Fig.~\ref{fig_inconsistancy_number}(c), these two cars are detected, resulting in two ``false negatives''. While the DNN correctly detects the left-most safety-relevant car in the real image, it misses it in the synthetic image, leading to a ``false positive'' in terms of using virtual testing to identify real-world safety issues. In other words, ``false negatives'' refer to the number of safety-relevant objects that are not detected in the real image but are detected in the synthetic images; Conversely, ``false positives'' represent the number of safety-relevant objects detected in real images but not in synthetic images. The sum of them is the total number of ``inconsistent predictions'' that we use for measuring SA-fidelity.

\begin{figure}[h!]
    \centering
\includegraphics[width=\columnwidth]{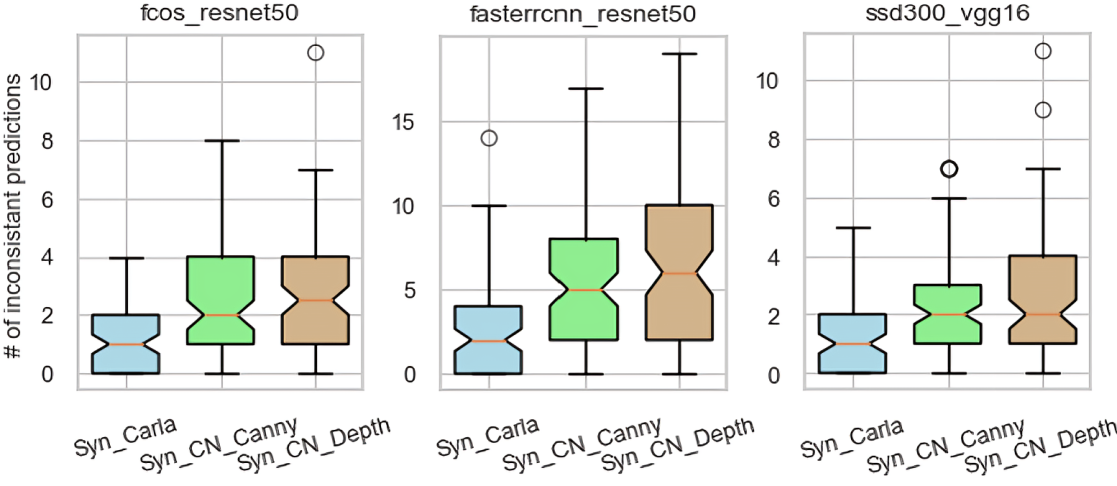}
 \caption{Statistics of inconsistent predictions, over 3 DNNs and 3 synthetic datasets.}
  \label{fig_rq1_boxplot}
  \vspace{-5mm}
\end{figure}

In Fig.~\ref{fig_rq1_boxplot}, we present the statistical results of inconsistent predictions, across~3 DNNs (one per column) and~3 synthetic datasets (indicated by boxes of~3 colors). In general, \txt{Syn\_Carla} is observed to have the highest SA-fidelity, while \txt{Syn\_CN\_Depth} appears to be the worst. Given such data, statistical hypothesis testing (e.g., T-test or U-test, depending on if the data distribution meets certain assumptions) can be employed to determine if one synthetic dataset is indeed significantly better than another in a principled way. We omit it for simplicity.

\subsubsection{RQ2} 
The first step is to identify configurable parameters $\theta$ in $\mathcal{C}_{\theta}$. For image-based inputs, apart from implementing $\mathcal{C}_{\theta}$ also as a DNN, one can also consider simpler implementations such as applying filters from the Python Imaging Library\footnote{ \url{https://pillow.readthedocs.io/en/stable/}}. In this case, $\theta$ can be parameters such as the degree of contrast, sharpness, brightness, or Gaussian blur to be additionally applied once when an image is generated. In RQ2, we adopt such a simple yet sufficient solution to conceptually validate the idea of calibrating $g_\theta$. Specifically, we conduct a grid search in the parametric space of 3 parameters representing contrast, sharpness, and brightness of filters. Each parameter ranges from~$0.8$ to~$1.2$ with a step size of~$0.1$.

The results of RQ2 are presented in Table~\ref{tab_calibration_rq2}. Despite we only searched for the optimal configuration in a very limited parametric space, we may already observe that: 
(1) The $\mathcal{C}_{\theta}$ indeed can be configured, while the best and worst configurations exhibit significant differences for SA-fidelity; 
(2) The optimal configuration of a given synthetic data generator varies across different DNNs;
(3) The best configuration for SA-fidelity is not necessarily the same as for OV-fidelity, nor is the worst configuration.

\begin{table*}[h!]
\caption{Best and worst configurations for SA-fidelity and OV-fidelity, given a DNN and synthetic data generator. NB, for each entry, the 3 parameters in bracket represent ``(contrast, brightness, sharpness)'' and the number after colon is the total number of inconsistent predictions over the calibration dataset.}
\label{tab_calibration_rq2}
\centering
\resizebox{0.83\textwidth}{!}{%
\begin{tblr}{
  cells = {c},
  cell{1}{3} = {c=2}{},
  cell{1}{5} = {c=2}{},
  cell{3}{1} = {r=3}{},
  cell{6}{1} = {r=3}{},
  cell{9}{1} = {r=3}{},
  vline{3} = {-}{},
  hline{3,12} = {-}{0.08em},
  hline{6,9} = {-}{dashed},
}
\hline
                                         &                & \uline{SA-Fidelity} &             & \uline{OV-Fidelity} &             \\
DNN                                      & dataset       & best conf.          & worst conf. & best conf.          & worst conf. \\
\begin{sideways}fcos\end{sideways}       & \texttt{Syn\_Carla}     & (1.2,1.2,0.8):105    & (0.8,0.9,1.1):123  &    (1.2,1.2,0.8):112    &  (0.9,0.9,1.1):131             \\
                                         & \texttt{Syn\_CN\_Canny} &  (0.8,0.8,0.8):216  & (0.9,1.2,0.9):231& (0.8,0.8,0.8):230    &  (1.0,1.2,0.8):247  \\
                                         & \texttt{Syn\_CN\_Depth} &   (0.8,0.8,0.8):230    &  (1.0,1.1,1.1):240 & (0.8,0.9,0.8):239 &(1.2,1.2,1.2):251 \\
\begin{sideways}fasterrcnn\end{sideways} & \texttt{Syn\_Carla}    & (0.9,1.2,0.8):215 & (1.0,0.8,1.1):244 &  (0.9,1.2,0.9):462   & (1.0,0.8,1.1):516 \\
                                         & \texttt{Syn\_CN\_Canny} & (0.8,0.8,0.8):475 &(1.1,1.2,0.9):512     & (0.8,0.8,1.2):873     & (1.2,1.2,1.0):960\\
                                         & \texttt{Syn\_CN\_Depth} & (1.1,0.8,0.8):592  & (0.9,0.9,0.9):610 & (0.8,0.8,0.9):958  & (1.2,1.2,1.1):992 \\
\begin{sideways}ssd300\end{sideways}     & \texttt{Syn\_Carla}     & (0.8,0.9,1.0):105  & (1.2,1.2,1.2):130 & (0.9,1.2,0.8):117 & (1.2,0.8,1.2):141 \\
                                         & \texttt{Syn\_CN\_Canny} &  (0.9,0.8,0.8):107  &  (1.0,1.2,1.2):117  &  (0.9,0.9,0.8):110  &  (1.0,1.2,1.2):120\\
                                         & \texttt{Syn\_CN\_Depth} &  (1.2,0.8,0.8):240 &  (1.2,1.1,0.9):256  &(1.2,0.8,0.8):252        &  (1.2,1.1,0.9):267
\end{tblr}
}
\end{table*}

\subsubsection{RQ3} In previous examples where style transfer is applied, the scenario description associated with an image is actually left implicit; we simply assume the described scenario to be the ground truth label translated into the detector output format (first split input image into grids of cells, and for each cell, consider the existence of a large object as an individual safety influencing attribute). 

In our last experiment, however, as we used DALLE-3 to generate ``semantically equivalent'' scenarios, we found it difficult to precisely control the number of vehicles and the location of the vehicles in a DALLE-3 synthesized image. This makes the implicit ``cell-based'' safety influencing attributes unusable, as for a synthetic-real image pair, the objectness in a particular cell may not be consistent, implying that even when the DNN makes no mistake, the ``cell-based'' safety-aware consistency by default does not hold. While one can relax by building a 0-1 attribute characterizing the existence of nearby objects, the relaxation also makes the optimization problem less useful due to abstracting the information too much. Ultimately, we used the \emph{consistency on false negative rate for close-by objects} as the optimization objective, which also resembles a \emph{generalized, probabilistic view} on the manifesting safety-critical concerns between a pair of synthetic and real images. Altogether, our assessment of the GAI-based synthetic data generator leads to interesting directions to be further explored.

\vspace{1mm}
\begin{remark}[\textbf{Generalisation to unseen scenarios}]
\label{remark_generlisation}
All~3 RQs are answered based on a ``calibration dataset'' containing paired real-world and synthetic data. In this dataset, each simulated scenario is either implicitly represented by the corresponding real-world data (for RQ1 and RQ2) or explicitly extracted from the real-world data through ``scenario understanding'' (for RQ3). A crucial underlying assumption, for practicality, is that both the assessment and optimal configuration of the synthetic data generator, derived from this calibration dataset, can be generalized to unseen or new scenarios (without collected real-word data) that will be simulated in the virtual testing. 
\end{remark}

While we plan to conduct rigorous experiments by splitting calibration and testing datasets in future works to validate such generalisablity, this first work mainly focus on the mathematical foundation and theoretical framework.

\section{Conclusion}
\label{sec:conclusion}

Our research underscores the critical importance of understanding and calibrating synthetic data to advance the safety and reliability of self-driving technology. By providing a cost-effective and scalable substitute to traditional real-world data collection methods, we have identified and addressed four distinct types of instance-level fidelity, extending beyond simple visual attributes. Our proposed SA-fidelity reflects a profound thinking on ``focusing on what really matters''. This approach is vital for aligning synthetic data with real-world safety challenges. We introduced an innovative optimization method to refine the synthetic data generator, effectively narrowing the fidelity gaps detected by the DNN-based component. Our findings demonstrate that this post-calibration improves the correlation between safety-critical errors in synthetic and real-world scenarios, marking a step forward in developing safer autonomous driving systems.

SA-fidelity reflects the practical need for synthetic and real image pairs to lead to similar safety concerns in a cost-effective manner. While such a concept is clear with motivations due to the irrelevance of, e.g., far objects, it is crucial to admit that we have not made explicit the precise definition of ``safety concerns", as our formulation is largely on perception systems. The use of output safety similarity as proposed in Def.~\ref{def:output.safety.similarity} suggests that a subset of semantic attributes in a scenario description is safety-critical, and it also assumes that \emph{the sameness of safety-relevant semantic attributes for two scenario descriptions implies that the autonomous system will make the same decision, thereby inducing the same safety concern when it exists}. Therefore, the importance of a well-designed scenario description language can not be underestimated. 

This work leads to many avenues to be explored. The first direction is to move beyond instance-level comparisons and also to extend the SA-fidelity concept to understand the closeness of synthetic data \textit{distribution} to the real-world scenario distribution, as well as accounting \textit{sequential dynamics} as the AV controller interacts with the virtual environments. Another direction is integrating the technique to enable the creation of an independent virtual testing and assessment platform on DNN-based systems.

%
%
%



 \addtolength{\textheight}{-2.1mm}   











\end{document}